\documentclass[conference,a4paper]{IEEEtran}

\usepackage{graphicx}
\usepackage{booktabs}
\usepackage{enumitem}
\usepackage{amsmath, amsthm, amssymb, esint}
\usepackage{epsfig}
\usepackage{color}
\usepackage{verbatim}
\newtheorem{thm}{Theorem}[section]
\newtheorem{definition}[thm]{Definition}

\newtheorem{prop}[thm]{Proposition}
\newtheorem{lemma}[thm]{Lemma}
\newtheorem{cor}[thm]{Corollary}

\newtheorem{rem}[thm]{Remark}

\newcommand{\E}{\mathbb{E}}

\newcommand{\R}{\mathbb{R}}

\newcommand{\Z}{\mathbb{Z}}

\newcommand{\vol}{\text{Vol}}

\newcommand{\Var}{\mathrm{Var}}

\newcommand{\diag}{\text{diag}}

\begin{document}



\sloppy 

\title{
{\color{black}
Lattice coding for Rician fading channels from Hadamard rotations}
}

\author{
  \IEEEauthorblockN{Alex Karrila, Niko R. V\"ais\"anen, David Karpuk, \emph{Member, IEEE}, and Camilla Hollanti, \emph{Member, IEEE}}
  \IEEEauthorblockA{Department of Mathematics and Systems Analysis\\
    Aalto University School of Science, Finland\\
    Email: firstname.(initial.)lastname@aalto.fi} 
}

\maketitle

\begin{abstract}
In this paper, we study lattice coding for Rician fading wireless channels. This is motivated in particular by preliminary studies suggesting the Rician fading model for millimeter-wavelength wireless communications. We restrict to lattice codes arising from rotations of $\Z^n$, and to a single-input single-output (SISO) channel. We observe that several lattice design criteria suggest the optimality of Hadamard rotations. For instance, we prove that Hadamard rotations maximize the diamond-packing density among all rotated $\Z^n$ lattices. Finally, we provide simulations to show that Hadamard rotations outperform optimal algebraic rotations and cross-packing lattices in the Rician channel.
\end{abstract}

\begin{keywords}
algebraic rotations, diamond packings, Hadamard rotations, lattice code design, orthogonal lattices, reliability, Rician fading, single-input single-output (SISO) channels
\end{keywords}

\section{Introduction}
\label{sec: intro}

Reliability is a key issue in designing wireless communications, since the channels are vulnerable to distortions. Reliability is typically improved by simultaneous error-correction coding and physical-layer design, with the tradeoff of decoding complexity and information rate.

Orthogonal lattice codes are a highly conventional physical-layer design in all types of wireless channels. Such codes provide for fast vector decoding based on solving a closest-vector problem. In addition, the Gray mapping from bit vectors to lattice vectors guarantees a beneficial conversion of vector decoding errors to bit decoding errors. Thus, the fundamental question in physical-layer reliability is to find the orthogonal lattice, or equivalently, the rotation, that provides a low rate of vector decoding errors. 
Algebraic
rotations are known to provide a solution in Rayleigh fading single-input single-output
(SISO), see \cite{Viterbo, Viterbo-taulukot} for a good overview, as well as in multiple-input multiple-output (MIMO)
channels \cite{Oggier}.


{\color{black}
The case is nevertheless not closed: research on future fifth generation (5G) communication networks calls for studying optimal rotations for the Rician channel.
Namely, instead of the usual 700MHz--2.6GHz, an extension to millimeter waves (e.g. 28GHz) is anticipated \cite{webmagazine1, webmagazine2}. Using millimeter waves provides both advantages and challenges. Using traditional spectrum allows for the transmission of data over a longer distance but at lower capacity, whereas millimeter wave offers greater bandwidth, but over shorter distances. However, the channel models for such new frequency spectra are not yet fully understood. For this reason, it is unclear at the moment what kind of modulation and encoding of the data will be most useful for energy efficiency and signal robustness. Nevertheless, tentative studies \cite{Rappaport} show that the scale channel coefficients are Rician distributed.}


In this paper, we show that several alternative design approaches suggest the optimality of \textit{Hadamard rotations} in Rician SISO channels. We show experimentally that Hadamard rotations outperform the algebraic rotations of \cite{Viterbo} and the cross-packing lattices of \cite{ViterboITW} over the Rician channel.  The Rician channel is indexed by a parameter $K>0$, with $K=0$ being the Rayleigh channel and $K \to \infty$ the Gaussian channel.  While the algebraic lattices offer better performance at $K = 0$, they are outperformed by the Hadamard rotations already at small $K$.  We present results only for $K = 20$ for the sake of compactness, but similar results were obtained for several $K$.

Lattices from Hadamard rotations has previously been proposed for certain fading channels in \cite{HadamardPrecoding1}, and as an alternative to OFDM for optical channels in \cite{HadamardPrecoding2}, to give a few applications. Nevertheless, it seems that their surprisingly good performance in the Rician case has not been noticed before.

\subsubsection*{Organization}
{\color{black}
In Section \ref{sec: preli}, we provide the necessary background on lattices, Hadamard rotations, and the Rician channel. The design approaches based on error probability estimates and sphere packings are presented in Sections \ref{sec: hadamard}--\ref{sec: SP}, and approaches based on diamond packings and diversity estimates in Section \ref{sec: connections}. Simulation results are provided in Section \ref{sec: simulations}.
}



\section{Preliminaries}
\label{sec: preli}

\subsection{Lattices}
\label{subsec: lattices}

A \textit{lattice} is a discrete additive subgroup of $\R^n$. We assume familiarity with the basic concepts related to lattices and lattice codes, and refer the unaccustomed reader to \cite{Viterbo} or \cite{Sloane}.  We point out that we consistently work in the column vector convention. We are going to be interested in the following class of lattices.

\begin{definition}
\label{def:WR}
A full-rank lattice $\Lambda\subset \R^n$ is \textit{well-rounded} (WR) if its minimal vectors span $\R^n$.
\end{definition}
The minimal vectors of a WR lattice $\Lambda$ are not guaranteed to generate $\Lambda$ \cite[Ch. 2]{Nguyen}. WR lattices are of interest here mainly due to their relation to the sphere-packing problem. Namely, all local maxima of the sphere packing density and hence in particular the sphere-packing optimal lattices are well-rounded: in a non-WR lattice, one can shrink the orthogonal complement of the minimal vectors to obtain a lattice with same minimal norm but smaller volume. The following is a partial converse to this statement:

\begin{lemma}
\label{lemma:ZnWR}
Let $\Lambda \subset \R^n$ be a WR full lattice, scaled to unit volume. Then, the sphere packing density of $\Lambda$ is minimized if and only if $\Lambda$ is a rotation of $\Z^n$.
\end{lemma}

\begin{proof}
Any $n$ linearly independent minimal vectors of $\Lambda$ generate a WR sublattice $\Lambda'$ whose sphere-packing density is smaller or equal to that of $\Lambda$. Thus, it suffices to prove the claim for the WR lattice $\Lambda'$ generated by its minimal vectors. But the claim is then immediate from Hadamard's inequality.
\end{proof}

\subsection{Hadamard matrices}

\begin{definition}
A \textit{(real) Hadamard matrix} is a square matrix whose all entries are $\pm 1$ and whose column vectors are orthogonal.
\end{definition}

The orthogonality condition can be equivalently cast as $W^T W = n I$. Thus, $W/\sqrt{n}$ is an orthogonal matrix, called a \textit{Hadamard rotation}. Hadamard matrices are conjectured to exist in all dimensions divisible by four, and known to exist in all such dimensions relevant for lattice coding purposes. The Kronecker product of two Hadamard matrices yields a third one. Based on this fact, \textit{Sylvester's construction} is the simplest way to obtain Hadamard matrices in dimensions that are powers of two, defined inductively by
\begin{equation}
W_2 = {\scriptstyle \begin{bmatrix}
1 & \hfill 1 \\
1 & -1
\end{bmatrix} } ,\quad W_{2^{k+1}} = W_2 \otimes W_{2^{k}}.
\end{equation}
We denote Hadamard matrices by $W$ and Hadamard rotations by $U$, often working with $W$ to avoid normalization constants.

\subsection{Fading SISO channels and the Rice distribution}
\label{subsec: channel models}

We consider a single-input-single-output (SISO) wireless channel.  We assume perfect channel state information (CSI) at the receiver and no CSI at the transmitter. Such a channel is modeled by the real channel equation
\begin{equation}
\label{eq: channel equation}
\mathbf{y} = \diag (\mathbf{h}) \mathbf{x} + \mathbf{v},
\end{equation}
where $\mathbf{x} \in \R^n$ and $\mathbf{y} \in \R^n$ are the transmitted and received vectors, respectively, and $\mathbf{h} \in \R_+^n$ and $\mathbf{v} \in \R^n$ are mutually independent random vectors modeling fading and noise, respectively.  We assume Gaussian noise, $\mathbf{v} \sim \mathcal{N}(\mathbf{0}, \sigma^2 I)$, and an interleaved channel where the $h_i$ are i.i.d.

{\color{black}
Our  primary interest lies in Rician distributed $h_i$, modeling a fading with a line of sight and scattering routes. The strength of the line of sight is captured by the \textit{Rician factor} $K \ge 0$ indexing the different Rice distributions defined by the density
\begin{equation}
\rho(h)=2h(1+K)e^{ -K-h^2(1+K)} 
I_0 (2h  \sqrt{ K^2+K } ), 
\end{equation}
where $I_0$ is the zeroth-order Bessel function of first kind. The case $K=0$ is the well-studied Rayleigh fading channel, while in the limit $K \to \infty$, $h_i$ becomes deterministically one and we obtain the additive white Gaussian noise (AWGN) channel. 
}

\section{Lattice Design Criteria in Rician Channels}
\label{sec: hadamard}
\label{subsec: heuristics}

In this section, we provide some motivating computations for the optimality of Hadamard rotations in physical-layer designs.  The design criteria 
rely on the $h_i$ having small variance, and thus being concentrated around their mean.  This is in contrast with the Rayleigh fading channel, in which deep fades (some $h_i\approx 0$) are a major cause of decoding errors.


The two most conventional ways to design physical-layer reliability in fading channels are minimizing the PEP bound \eqref{eq: PEP} or, in the AWGN channel, maximizing the sphere-packing density of the lattice.
We study analogous design approaches in channel models where deep fades are not the primary cause of decoding errors, such as the Rician channel\footnote{This assumption is validated for the Rician channel in Section \ref{sec: SP}.}, by studying the pairwise error probability (PEP) estimate and the sphere-packing at a near-average fade. We find an agreement of the two approaches, both suggesting the optimality of Hadamard rotations. 

\subsection{Pairwise error probability}
 The standard PEP bound states that the probability $P$ of a vector decoding error is bounded by
 {\color{black}
\begin{equation}
\label{eq: PEP}
P \le \frac{1}{2}  \sum_{\mathbf{t} \in \Lambda \setminus \{ \mathbf{0} \} }\mathbb{E} \left\{ \exp \left( - \frac{\Vert \diag(\mathbf{h}) \mathbf{t} \Vert^2  }{8 \sigma^2 } \right) \right\} ,
\end{equation}
where $\Lambda$ is the code lattice, $\sigma^2$ the noise variance, and the expectation is over $\mathbf{h}$. In the deep-fade dominated Rayleigh fading channel, this expectation is analyzed by fixing the distribution of the fading $h_i$ and expanding around $\sigma^2 = 0$. To study a noise dominated channel, let us fix the noise $\sigma^2$ and ``expand around $\Var(h^2_i) = 0$'':
denoting $(\Vert \diag(\mathbf{h}) \mathbf{t} \Vert^2   - \mathbb{E} \{ \Vert \diag(\mathbf{h}) \mathbf{t} \Vert^2 \} )/8 \sigma^2 = \epsilon$, the exponential in \eqref{eq: PEP} becomes
\begin{align}
\nonumber
& \exp  \left( - \frac{\Vert \diag(\mathbf{h}) \mathbf{t} \Vert^2  }{8 \sigma^2 } \right) = \exp \left( - \frac{\mathbb{E} \{ \Vert \diag(\mathbf{h}) \mathbf{t} \Vert^2 \} }{8 \sigma^2 } \right) e^{-\epsilon}\\
&= \exp \left( - \frac{\mathbb{E} \{ \Vert \diag(\mathbf{h}) \mathbf{t} \Vert^2 \} }{8 \sigma^2 } \right) (1 - \epsilon + \epsilon^2/2 - \ldots) 
\end{align}
Neglecting the higher-order terms represented by the ellipses and substituting $\mathbb{E} \{ \Vert \diag(\mathbf{h}) \mathbf{t} \Vert^2 \} = \mathbb{E} \{ h^2 \}\Vert \mathbf{t} \Vert^2$, we approximate the PEP bound as}
{\color{black}
\begin{eqnarray}
&& \frac{1}{2} \sum_{\mathbf{t} \in \Lambda \setminus \{ \mathbf{0} \} } \mathbb{E} \left\{  \exp \left( - \frac{\Vert \diag(\mathbf{h}) \mathbf{t} \Vert^2  }{8 \sigma^2 } \right) \right\} \\
& \approx &
\frac{1}{2} \sum_{\mathbf{t} \in \Lambda \setminus \{ \mathbf{0} \} } e^{ -\frac{\mathbb{E} \{ h^2 \}\Vert \mathbf{t} \Vert^2}{8 \sigma^2} }\left[ 1 + \frac{\Var (h^2) }{2(8 \sigma^2 )^2}\Vert \mathbf{t} \Vert_4^4\right],
\end{eqnarray}
where $\Vert \cdot \Vert_p$ denotes the usual $L^p$ vector norm.
Since $\Var (h^2)$ was assumed small,  the minimal vectors of $\Lambda$ dominate the series above, and  their $L^2$ norm should be maximized and $L^4$ norm minimized in order to minimize the error probability. }In other words, we should first maximize the packing density of $\Lambda$ and then rotate it so that the minimal vectors are parallel to $[\pm 1, \ldots, \pm 1]^T$. If $\Lambda$ is a rotation of $\Z^n$ this condition is satisfied if and only if it is a Hadamard rotation.

\subsection{Sphere packings}
Regarding the fading channel as an instantaneous Gaussian channel, we should maximize the packing density of the randomly faded lattice $\diag (\mathbf{h}) \Lambda$. First, the average norm of a given lattice vector $\mathbf{t} \in \Lambda$ is after fading $\mathbb{E} \{ \Vert \diag(\mathbf{h}) \mathbf{t} \Vert^2 \} = \mathbb{E} \{ h^2 \}\Vert \mathbf{t} \Vert^2$.
This tells us to maximize the packing radius of $\Lambda$, but does not differentiate between rotations of $\Lambda$. Next, the random norms $\Vert \diag( \mathbf{h} )  \mathbf{t}  \Vert^2$, especially the shortest ones, should be stabilized around their expectation $\E [ \Vert \diag(\mathbf{h})  \mathbf{t}  \Vert^2 ] $. Hence, we should minimize the variance
\begin{equation}
\Var ( \Vert \diag(\mathbf{h})  \mathbf{t}  \Vert^2 / \E [ \Vert \diag(\mathbf{h})  \mathbf{t}  \Vert^2 ]  ) 
= \frac{\Var(h_i^2)}{\E[h_i^2]^2} \frac{\Vert \mathbf{t}  \Vert_4^4}{\Vert \mathbf{t}  \Vert_2^4}.
\end{equation}
The conclusions are identical to those reached from the criteria derived from the PEP; specifically, if $\Lambda$ is a rotation of $\Z^n$ and $\mathbf{t}$ a minimal vector of $\Lambda$, then the above quantity will be minimized exactly when it is a Hadamard rotation.

\begin{rem}
More generally, we could expand the class of lattices we are interested in to include non-orthogonal lattices.  In particular, there exist many well-rounded lattices (which necessarily have good sphere packings) all of whose minimal vectors are  parallel to $[ \pm 1, \ldots, \pm 1]^T $.  For example, the body-centered cubic lattice in $\R^3$, generated by the vectors $[ 1, 1, 1]^T $, $[ 1, -1, -1]^T $, and $[ -1, -1, 1]^T $, or its tensor product with any Hadamard lattice.
\end{rem}

\section{Sphere-packing density of Hadamard lattices in Rician fading channels}
\label{sec: SP}

In this section, we provide a probabilistic estimate for the sphere-packing density of a Hadamard rotated unit lattice after fading. We work with the unnormalized Hadamard matrices $W \in \R^{n \times n}$, and denote the lattices of interest by $\Lambda = W \Z^n$ and $\Lambda_h = \diag(\mathbf{h}) W \Z^n$.

Notice that the natural generators of the faded lattice $\Lambda_h$ are of the form $[\pm h_1 \ldots , \pm h_n]^T$. If they are minimal vectors of $\Lambda_h$, then $\Lambda_h$ is well-rounded and, by Lemma \ref{lemma:ZnWR}, $\Lambda_h$ has a good sphere packing. 
In this section, we compute the probability of this event in low dimensions.

\begin{lemma}
\label{lemma: shortest vectors}
For any realization of $\mathbf{h}$, there is a minimal vector $\diag (\mathbf{h}) W \omega$ of the faded lattice $\Lambda_h$, where the integer lattice coordinates $\omega \in \Z^n$ either satisfy $\Vert \omega \Vert^2 < n$, or $\omega$ is a row of the 
 matrix $W$.
\end{lemma}

\begin{proof}
Let $w_j$ be the $j^{th}$ row of $W$, and $\omega \in \Z^n$ some arbitrary lattice coordinates. We compare the lengths of the lattice $\Lambda_h$ vectors $\diag(\mathbf{h}) W w_j$ and $\diag(\mathbf{h}) W \omega$. First, $ W w_j = n \mathbf{e_j}$, so for some $j$, we have 
\begin{equation}
\Vert \diag (\mathbf{h}) W w_j \Vert^2 = n^2 \min_{1 \le i \le m} h_i^2.
\end{equation}
Next, denote $W \omega = \mathbf{z}$ and note that 
\begin{eqnarray}
\Vert \diag (\mathbf{h}) W \omega \Vert^2 &=& \sum_{k = 1}^n h_k^2 z_k^2 \\
&\ge & \Vert \mathbf{z} \Vert^2 \min_{1 \le k \le n} h_k^2 \\
\label{general}
& = & n \Vert \omega \Vert^2 \min_{1 \le k \le n} h_k^2,
\end{eqnarray}
where the last step used the fact that $W/\sqrt{n}$ is a rotation matrix. Now $\diag (\mathbf{h}) W \omega$ must satisfy $\Vert \omega \Vert^2 < n$ to be shorter than $\diag(\mathbf{h}) W w_j$.
\end{proof}


\begin{cor}
\label{cor: finite dimensions}
Let us denote by $C$ the event that the natural generators $\diag(\mathbf{h})W\mathbf{e}_i$ of a faded Hadamard lattices $\Lambda_h = \diag(\mathbf{h}) W \Z^n$ are minimal vectors. Then,
$\mathbb{P} \{ C \}$ is given by integrating the joint density of $h_ 1^2,...,h_n^2$ over a cone.
\end{cor}

\begin{proof}
For fixed dimension, the previous lemma gives finitely many lattice coordinates $\omega$ that can yield minimal vectors. Then, $\diag (\mathbf{h}) W \mathbf{e}_j$ are minimal vectors if and only if all such $\omega$ satisfy
\begin{eqnarray}
\Vert \diag (\mathbf{h}) W \omega \Vert^2 &\ge & \sum_{k = 1}^n h_k^2 ,
\end{eqnarray}
a linear inequality in $h_1^2,...,h_n^2$.
\end{proof}

\begin{figure}
\includegraphics[width = 0.45\textwidth]{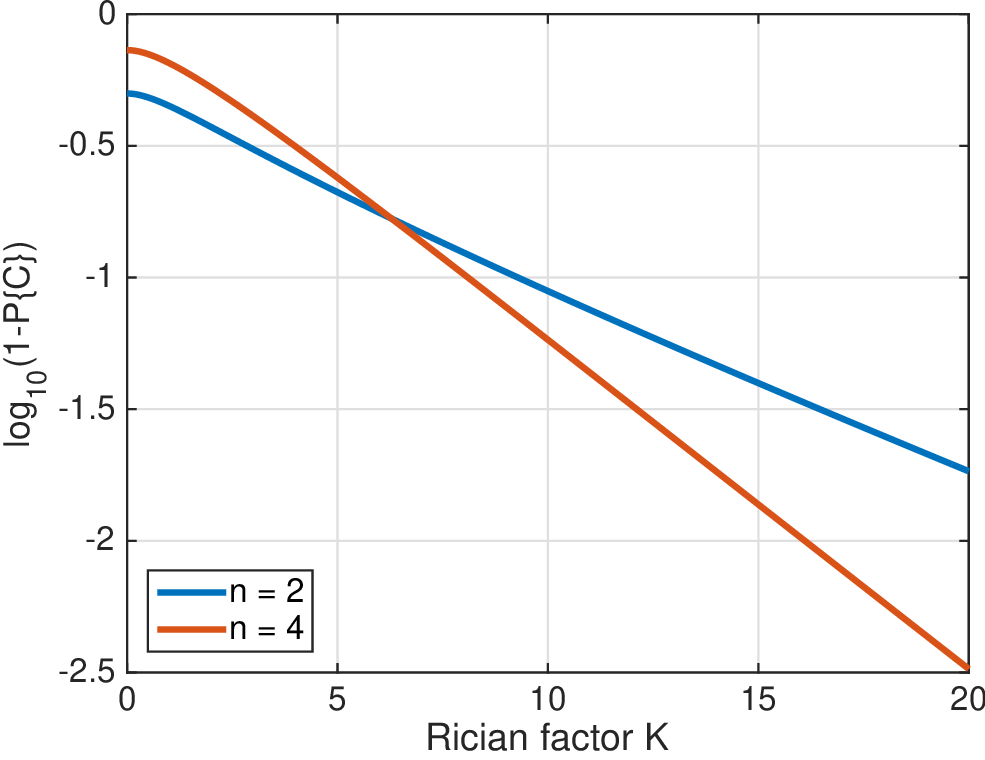}
\caption{The probability $(1- \mathbf{P} \{ C \})$  of Corollary \ref{cor: finite dimensions}, when $h_i$ are taken Rician with parameter $K$,  for the Hadamard matrices in $n = 2,4$ dimensions from Sylvester's construction.
\label{fig: non WR proba}}
\end{figure}

This corollary allows for numerical computations, as illustrated in Figure \ref{fig: non WR proba}. In particular, when $h_i$ are taken Rician with parameter $K$,  the probability  $(1- \mathbf{P} \{ C \})$ of $\Lambda_h$ not being WR decays exponentially in $K$. In four dimensions, the exponential decay is faster than in two. This leads to two conclusions. First, especially in larger dimensions, one can expect Hadamard rotations to perform well in Rician channels where $K$ is sufficiently away from zero. Second, the decoding errors in such a setup mainly occur due to large noise, rather than deep fades.  
This supports the standing assumption of ``noise-dominated errors'' in the design approach suggesting Hadamard rotations in Section \ref{subsec: heuristics}.

\section{Connections to Other Design Criteria}
\label{sec: connections}

Aside from the above criteria derived from the PEP and sphere packing density, there have recently been more subtle design criteria introduced for lattice codes over fading channels, wherein deep fades are not the primary cause of errors.  These include the cross-packing density \cite{ViterboITW} and local diversity \cite{Karpuk-Hollanti}.  In this section, we present results which study how well Hadamard rotations of $\Z^n$ satisfy these design criteria.

\subsection{Diamond packings} Recently, \cite{ViterboITW} designed lattices for Rician channels by maximizing their \textit{cross-packing} density (equivalently, maximizing cross-packing radius), where the crosses consist of axis-aligned line segments. This was motivated by approximating the shape of a contour surfaces of the terms in the Rician PEP estimate \eqref{eq: PEP} by a cross polyomino, i.e.\ an appropriately thickened cross. Analogously, approximating the shape of the contour surfaces by the convex hulls of the $n$-dimensional crosses, also known as $L^1$ norm balls or diamonds, one ends up designing lattices based on their diamond packing density. We find that the Hadamard rotations maximize the diamond packing density of the rotated $\Z^n$ lattices.

\begin{prop}
Let  $R \in \R^n$ be a rotation matrix. Then, the minimal $L^1$ norm of the rotated $\Z^n$ lattice $R \Z^n$ satisfies
\begin{equation}
\min_{\substack{\mathbf{t} \in R\Z^n \\ \mathbf{t} \ne \mathbf{0} }} \Vert \mathbf{t} \Vert_1 \le \sqrt{n},
\end{equation}
with equality if and only if $R$ is a Hadamard rotation.
\end{prop}

\begin{proof}
To prove the inequality, let $\mathbf{e}$ be any elementary basis vector. Take $\mathbf{t} = R \mathbf{e}$ a rotated basis vector. Recall now the relation of $L^2$ and $L^1$ norms on $\R^n$:
\begin{equation}
\Vert \mathbf{t} \Vert_1 \le \sqrt{n} \Vert \mathbf{t} \Vert_2,
\end{equation}
with equality if and only if $\mathbf{t}$ is parallel to $[\pm 1, \ldots, \pm 1 ]^T$.
This implies that 
\begin{equation}
\min_{\substack{\mathbf{t} \in R\Z^n \\ \mathbf{t} \ne \mathbf{0} }} \Vert \mathbf{t} \Vert_1 \le \sqrt{n},
\end{equation}
and the equality is possibly reached only if $R$ is Hadamard.

To prove that the inequality is sharp for Hadamard rotations, notice that $\Vert \mathbf{t} \Vert_1 \ge a$ if and only if $ \mathbf{t} \cdot \mathbf{s} \ge a$ for some sign vector $\mathbf{s} = [\pm 1, \ldots, \pm 1 ]^T$. Now, let $\mathbf{t}$ be any nonzero vector of a Hadamard rotated lattice $U\Z^n$, and let $\mathbf{u}$ be a basis vector of the lattice $U\Z^n$ (i.e., a column of $U$) with a nonzero (integer) coefficient in $\mathbf{t} \in U\Z^n$. Choosing $\mathbf{s} = \pm \sqrt{n} \mathbf{u}$, the orthogonality of the Hadamard basis implies
\begin{equation}
\mathbf{t} \cdot \mathbf{s} \ge \sqrt{n},
\end{equation}
so indeed $\Vert \mathbf{t} \Vert_1 \ge \sqrt{n}$ for all nonzero $\mathbf{t} \in U\Z^n$.
\end{proof}

\subsection{Local diversity} In \cite{Karpuk-Hollanti} the design of reliable lattices in low signal-to-ratio (SNR) Rayleigh fading channels was considered, and the authors deduced that Hadamard rotations are optimal within a certain one-parameter Lie group of rotations. They explained the appearance Hadamard rotations, which contrary to conventional algebraic rotations are not fully diverse, by \textit{local diversity} of Hadamard lattices, \textit{i.e.}, a tradeoff between diversity and length of the lattice vectors. We prove the following sharp local diversity estimate for Hadamard rotations. 

\begin{prop}
\label{lem: local diversity}
Let $U \in \R^{n \times n}$ be a Hadamard rotation and let $\mathbf{t} = U \omega$ be of diversity $k > 0$. Then,
\begin{equation}
\label{local diversity}
k \Vert \mathbf{t} \Vert^2 \ge n.
\end{equation}
\end{prop}

\begin{proof}
Since $\mathbf{t}$ has diversity $k$, we have $\Vert \mathbf{t} \Vert_2 \ge \Vert \mathbf{t} \Vert_1 /  \sqrt{k} $. Substituting the minimal $L^1$ norm $\Vert \mathbf{t} \Vert_1 \ge \sqrt{n}$, we obtain
$
\Vert \mathbf{t} \Vert_2^2 \ge n / k
$
as desired.
\end{proof}

\section{Simulations}
\label{sec: simulations}

We now present simulation results to confirm our previous findings. The key parameters for our simulations are the dimension $n$ of the signal constellation and its order $q$ per dimension (so that the number of constellation points is $q^n$), the parameter $K$ of the Rician distribution, and the volume-to-noise ratio (VNR) which defines the variance of the noise $\sigma_n^2$ by the formula $\mathrm{VNR} = \vol(\Lambda)^{2/n}/(8\sigma_n^2)$. In simulations for $n = 2$ we used $q=8$, and for $n = 4$ we used $q=4$.

\subsection{Setup details}

Let $M$ denote the generator matrix of the simulated code lattice $\Lambda$ and $S$ the signal constellation.  To construct the signal constellation, we start from a finite region of $\mathbb Z^n$ described by $S'' = \{ (x_1, \dots, x_n)\ |\ x_i \in \mathbb Z,\ 0 \le x_i < q\ \forall i \}$. Then, we center such a region by setting $S' = S'' - (q-1)/2$. Finally, our signal constellation $S$ is a image of $S'$ under the generator matrix, $S = MS'$.  

The channel simulations were based on first generating a uniform random constellation vector $\mathbf{x}$, then generating the Rician fading and Gaussian noise vectors of the channel equation \eqref{eq: channel equation} to obtain the received vector $\mathbf{y}$, and finally solving the closest vector problem $\hat { \mathbf{x} } = \arg \min_{\mathbf{t} \in S} ||\mathbf{y} - \mathrm{diag}(\mathbf{h})\mathbf{t}||^2$. The decoding is correct if $\hat { \mathbf{x} } = \mathbf{x}$. 

\subsection{Results}

\subsubsection*{Rotations}

We first compared three different rotations of $\Z^n$ in dimension $n = 4$. These rotations are the identity rotation, the best known algebraic rotation \cite{Viterbo-taulukot}, and the Hadamard rotation from Sylvester's construction. We simulated the performance of these rotations in 
the Rician channel over a large range of $K$ values.

\begin{figure}
\centering
\includegraphics[width=0.45\textwidth]{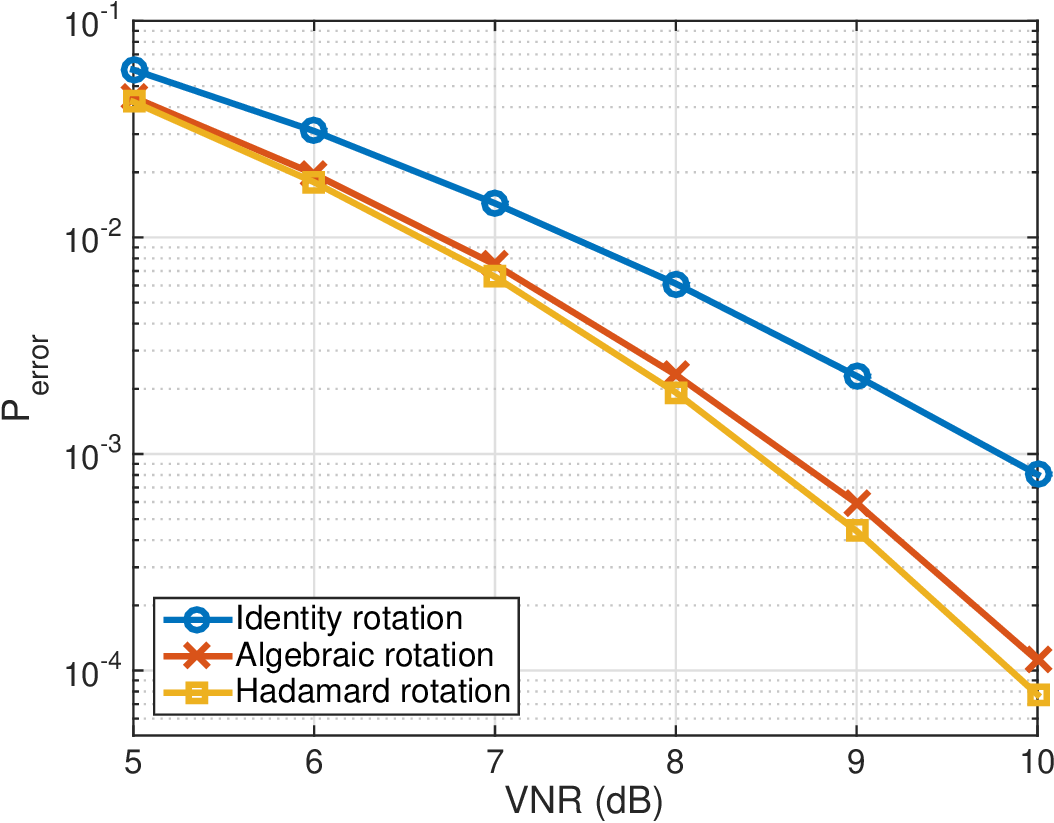}
\caption{Error rates in the Rician fading channel with parameter $K=20$ for lattice codes from the identity rotation, the best-known algebraic rotation, and the Hadamard rotation in $n = 4$ dimensions. \label{fig: hada_vs_alg}}
\end{figure}

In Fig.\ \ref{fig: hada_vs_alg} we plot error rate as a function of VNR, for the Rician parameter $K=20$.  The Hadamard rotations perform slightly better than the algebraic rotations over the whole VNR range, which supports the results of Section \ref{sec: hadamard}.  Similar simulations for other values of $K$ produced comparable results.  

\begin{figure}
\centering
\includegraphics[width=0.45\textwidth]{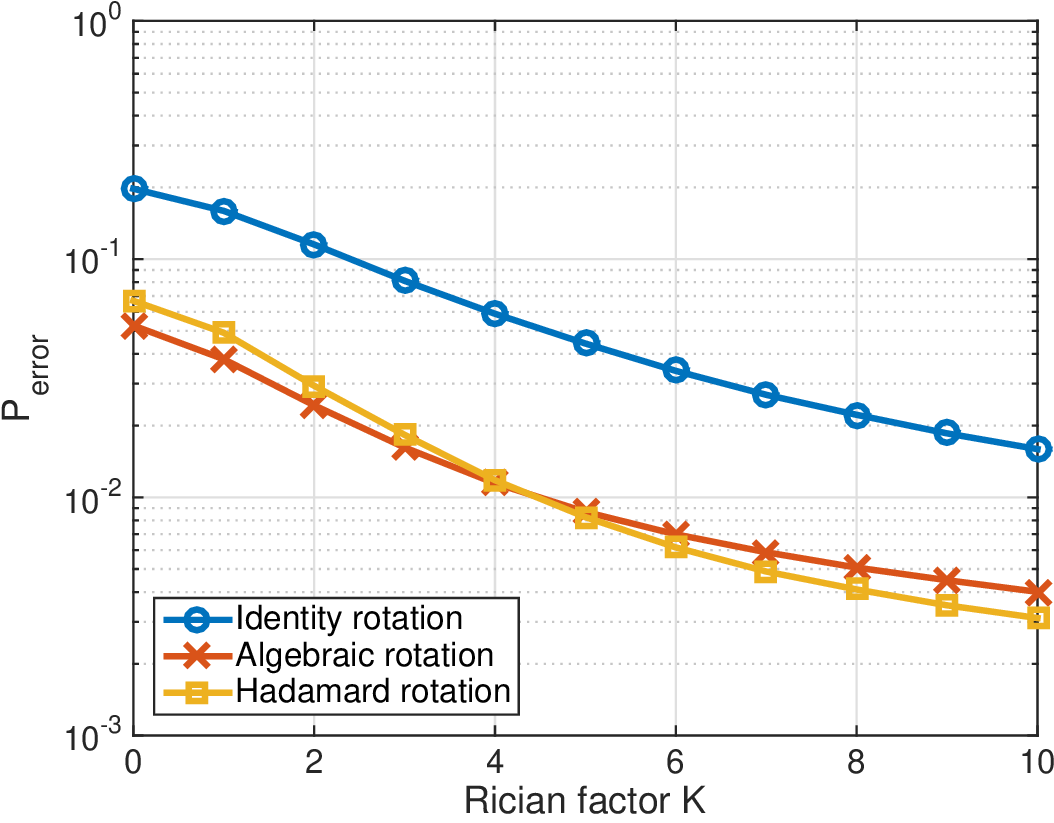}
\caption{Error probabilities for the Rician fading channel in $n = 4$ dimensions as a function of the parameter $K$, with fixed VNR $=8$~dB. \label{fig: k_vary}}
\end{figure}

To see how the performance of the different lattices varies with the Rician parameter $K$, we plot in Figure \ref{fig: k_vary} the error rates as a function of $K$ for fixed VNR.  We see that in dimension $n = 4$ there is a value for $K$, namely $K\approx 4.4$, after which the Hadamard rotation performs better than the other two simulated rotations.  An analogous simulation in dimension $n = 2$ produced similar results, with the critical value of $K$ being $K\approx 7.2$.


\subsubsection*{Cross-packing lattices}

\begin{figure}
\centering

\includegraphics[width=0.45\textwidth]{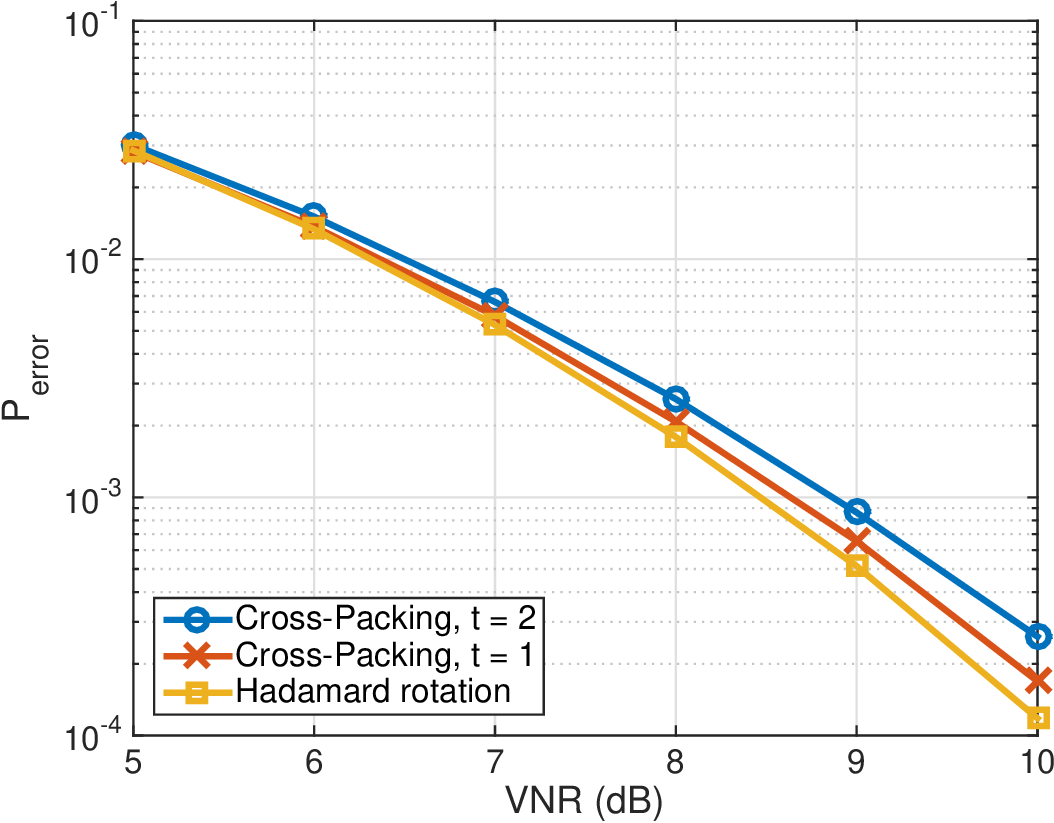}
\caption{Comparison of Hadamard rotation and cross-packing lattices in dimension $n = 2$, for Rician factor $K = 20$. \label{fig: hada_vs_cross}}
\end{figure}

Recently, \cite{ViterboITW} designed a family of cross-packing lattices for Rician channels, indexed by a parameter $t$, of which $t=1$ and $t=2$ performed best. Let us compare the Hadamard rotations against these lattices (normalized to unit volume). 
Note that the cross-packing lattices are not orthogonal, so comparing vector error rates in the codes, as plotted, is not equivalent to comparing bit error rates. 

In Fig.\ \ref{fig: hada_vs_cross} we present simulation results which compare the performance of the Hadamard rotation to the cross-packing lattices in dimension $n = 2$, which show that the Hadamard rotation offers a modest improvement over the cross-packing lattices.


\section{Conclusions}

Motivated by applications to 5G and millimeter wave communications, we studied lattice codebook design for the Rician fading channel.  It was found that Hadamard rotations of $\Z^n$ satisfy the design criteria derived from the corresponding PEP and the sphere-packing density.  Two particularly attractive features of Hadamard lattices are that they often retain good sphere-packing properties after experiencing fading, and that they are maximizers of the diamond packing density among rotations of $\Z^n$. Simulations were provided which demonstrate that Hadamard lattices outperform other lattice constructions in the Rician channel, such as algebraic rotations and cross-packing lattices.



\end{document}